\documentclass[letterpaper, 10 pt, conference]{ieeeconf}
\IEEEoverridecommandlockouts
\usepackage{cite}
\usepackage{amsmath,amssymb,amsfonts,amsthm,color,float}
\usepackage{algorithm}
\usepackage{algpseudocode}
\usepackage{graphicx}
\usepackage{textcomp}
 \usepackage{calc}
 \usepackage{tikz}
\usepackage{cases}
\usepackage{verbatim}
\usepackage{hyperref}
\usepackage{xcolor, soul}
\usepackage{nccmath}
\usepackage{mathtools, nccmath}
\usepackage{caption}
\usepackage{subcaption}

\usetikzlibrary{arrows}
\usetikzlibrary{shapes}
 \usetikzlibrary{calc}
\usetikzlibrary{decorations.pathreplacing}
\usetikzlibrary{arrows,positioning}

\theoremstyle{plain}
\newtheorem{thm}{Theorem}

\newtheorem{prop}{Proposition}

\theoremstyle{definition}

\newtheorem*{rem*}{Remark}

\DeclareMathOperator*{\argmax}{arg\,max} 
\newcommand{\floor}[1]{\left \lfloor #1 \right \rfloor}
\newcommand{\ceil}[1]{\left \lceil #1 \right \rceil}
\newcommand{\card}[1]{| #1 |}

\newcommand{\ac}{\mathcal{A}} 
\newcommand{\p}{\mathcal{I}} 
\newcommand{\W}{\mathrm{W}} 
\newcommand{\emp}{a^{\varnothing}}  
\newcommand{\aopt}{a^{\mathrm{opt}}}  
\newcommand{\pathf}{p_{i \to i+1}}  

\newcommand{\gr}{\mathcal{G}} 
\newcommand{\vr}{\mathcal{V}} 
\newcommand{\ed}{\mathcal{E}} 
\newcommand{\nd}{\mathcal{N}} 
\newcommand{\gi}{\tilde{a}} 
\newcommand{\gt}{T(\gr, \pi)} 
\newcommand{\gtw}{T_{\max}(\gr)} 
\newcommand{\gtb}{T_{\min}(\gr)} 
\newcommand{\talg}{T_{\mathrm{alg}}(\gr, v_{\mathrm{seed}})} 
\newcommand{\T}{\mathcal{T}} 
\newcommand{\pib}{\pi_{\mathrm{best}}} 
\newcommand{\piw}{\pi_{\mathrm{worst}}} 

\begin{document}
\title{Execution Order Matters in Greedy Algorithms with Limited Information}

\author{Rohit Konda, David Grimsman, Jason R. Marden \thanks{R. Konda and J. R. Marden are with the Department of Electrical and Computer Engineering at the University of California, Santa Barbara, CA and D. Grimsman is with the Department of Computer Science at Brigham Young University. This work is supported by \texttt{AFOSR Grant \#FA9550-20-1-0054} and \texttt{AFOSR Grant \#FA9550-21-1-0203}.}}

\maketitle
\thispagestyle{empty}

\begin{abstract}
In this work, we study the multi-agent decision problem where agents try to coordinate to optimize a given system-level objective. While solving for the global optimum is intractable in many cases, the \emph{greedy algorithm} is a well-studied and efficient way to provide good approximate solutions - notably for submodular optimization problems. Executing the greedy algorithm requires the agents to be ordered and execute a local optimization based on the solutions of the previous agents. However, in limited information settings, passing the solution from the previous agents may be nontrivial, as some agents may not be able to directly communicate with each other. Thus the communication time required to execute the greedy algorithm is closely tied to the order that the agents are given. In this work, we characterize interplay between the communication complexity and agent orderings by showing that the complexity using the best ordering is $O(n)$ and increases considerably to $O(n^2)$ when using the worst ordering. Motivated by this, we also propose an algorithm that can find an ordering and execute the greedy algorithm quickly, in a distributed fashion. We also show that such an execution of the greedy algorithm is advantageous over current methods for distributed submodular maximization.
\end{abstract}

\section{Introduction}
\label{sec:intro}

Many real-world problems are well-modeled as multiagent decision problems, including building energy management \cite{Zhao2013}, stock trading \cite{Luo2002}, water resource allocation \cite{Kahil2016,Madani2010}, traffic light management \cite{Cruz-Piris2018}, the power grid \cite{Bai2006}, and robot path planning \cite{SinghPal2013}. In these scenarios, the set of $n$ decision makers, or agents, coordinate to a joint decision that maximizes some objective function.

In general, finding the optimal decision set is computationally intractable, even for a centralized authority. Therefore, there exist a multitude of techniques for arriving at a joint decision, which may be an approximation of the optimal. For instance, consensus algorithms offer a way for agents to converge as a group toward a unified decision \cite{Nguyen2018,Kar2010,Ongaro2019}. In other settings, a game-theoretic approach is advantageous, where agents arrive at a joint decision which is some form of equilibrium (e.g., Nash equilibrium \cite{nash1950equilibrium}, Wardrop equilibrium \cite{wardrop1952road}, Stackelberg equilibrium \cite{von2010market}, etc.). Of course even finding such equilibria can be challenging \cite{gilboa1989nash}, but there are subclasses of problems where this can be done efficiently \cite{monderer1996potential}.

Another common approach to multiagent decision problems is a greedy algorithm \cite{Marzouki2017,Gharesifard2016a}. A common theme among greedy algorithms is that at each iteration of the algorithm, a myopic choice is made: simply pick the best immediate option, ignoring the effect on future iterations of the algorithm. As with the other algorithms mentioned above, greedy algorithms in general are not always guaranteed to find an optimal solution to a given problem, however, they are often easy to implement, execute quickly, and in some cases provide some degree of optimality.

This work focuses on the scenario where a greedy algorithm is used to solve a multiagent decision problem. In this setting, a greedy algorithm is implemented by first ordering the agents. Then each agent sequentially makes its decision by choosing the action that maximizes the objective function, based solely on the decisions of previous agents in the sequence. An underlying element of the greedy algorithm is that the agents are able to coordinate with each other via some network. In the best case, such a network would allow for each agent to communicate with all other agents directly. In many applications, however, this is not realistic; communication between agents $i$ and $j$ must pass through other agents in the network. If $i$ and $j$ are on opposite ends of the network, or if the network has highly-limited bandwidth, this communication may be delayed. In light of this, two questions arise:
\begin{enumerate}
    \item Given the structure of the communication network, how does the ordering of the agents affect the time it takes to complete the greedy algorithm?
    \item Can the agents coordinate among themselves to find the ordering that will cause the greedy algorithm to complete as fast as possible?
\end{enumerate}
We address the first question by showing, given a network structure, that the greedy algorithm finishes in $O(n^2)$ time steps for the worst ordering and $O(n)$ time steps for the best ordering. We then address the second question by presenting a fully-distributed algorithm whereby agents can find a near-optimal ordering while simultaneously runnning the greedy algorithm. 

Of particular import in this work are submodular maximization problems, which are prevalent in modeling many applications, such as sensor placement~\cite{Krause2008}, data summarization~\cite{Badanidiyuru2014,Lin2011}, robot path planning~\cite{Singh2007,Corah2019}, task allocation~\cite{Arslan2007}, inferring influence in a social network~\cite{Gomez-Rodriguez2012}, image segmentation~\cite{Kim2011a}, outbreak detection in networks~\cite{Leskovec2007}, and leader selection in multiagent systems~\cite{Clark2011}. A key feature that is shared among the objective functions in these various domains is a property of \emph{diminishing returns}. For example, in outbreak detection, the added benefit of placing an outbreak sensor on a node in a network is valuable when there are few other sensors in the network, and less valuable when there are already many other sensors present. Objectives that exhibit such properties are \emph{submodular}.

While such problems are NP-Hard in general, the property of submodularity can be exploited to show that certain algorithms can achieve near-optimal results. The seminal work in \cite{nemhauser1978analysis} shows that a centralized greedy algorithm can, in fact, provide a solution that is guaranteed to be within $1/2$ of the optimal solution. More sophisticated algorithms have pushed this guarantee from $1/2$ to $1 - 1/e \approx 0.63$~\cite{Calinescu2011,Filmus2012a}. Progress beyond this approximation frontier is not possible for polynomial time algorithms as it was also shown that no such algorithm can achieve higher guarantees than $1 - 1/e$, unless $P=NP$~\cite{Feige1998}.

Recently, work has emerged wherein submodular maximization problems are modeled as multiagent decision problems \cite{mirzasoleiman2016,Robey2019a,rezazadeh2021distributed,du2020jacobi}. In Section~\ref{sec:AA}, we will show how our version of the greedy algorithm applies to this setting \footnote{It should be noted that the work in \cite{grimsman2017impact} also explores the idea of performing the standard greedy algorithm on a network for a submodular maximization problem. It was shown that as links are removed from the network, the performance guarantee decreases. The setting of our current paper differs in that we allow $k$-hop communication, which is why the $1/2$ guarantee is maintained. It also differs in that the work in \cite{grimsman2017impact} assumes an implicit ordering of the agents, whereas in this paper the agents must coordinate to find an ordering.}, and how this algorithms compares to existing techniques in terms of runtime and performance guarantees. It will be shown that the greedy algorithm will complete in fewer time steps than existing methods, while still maintaining 1/2-optimality in the resulting decision set.

In Section \ref{sec:model}, we introduce our model. In Section \ref{sec:convergence}, we present our main results on communication time guarantees versus different orderings. We empirically verify our theoretical results in Section \ref{sec:sim} and discuss the implications in submodular maximization problems in Section \ref{sec:AA}. We conclude in Section \ref{sec:conc}. The relevant code is found at \cite{konda2021}. We sincerely thank Gilberto Diaz-Garcia for the helpful discussions.

\section{Model}
\label{sec:model}

Consider a distributed optimization problem with $n$ agents $\p=\{1, \dots, n\}$, where each agent is endowed with a decision or action set $\ac_i$. We denote an action as $a_i \in \ac_i$, and a joint action profile as $a \in \ac = \ac_1 \times \cdots \ac_n$. We assume that each agent $i$ has the ability to ``opt out" of participating in the decision process. This is modeled by having an action $\emp_i \in \ac_i$, so that when agent $i$ chooses action $\emp_i$, agent $i$ is opting out. The quality of each joint action profile is evaluated with a global objective function $\W(a): \ac \to \mathbb{R}_{\geq 0}$ that a system designer seeks to maximize. In other words, the goal of the system designer is to coordinate the agents to a joint action profile that satisfies 

\begin{equation} \label{eq:opt}
    \aopt \in \argmax_{a \in \ac} \W(a).
\end{equation}

In general, solving the multi-agent decision problem in Eq. $\eqref{eq:opt}$ is infeasible, due to computational, informational, communication constraints etc. Therefore, fast, distributed algorithms are employed to compute good approximate solutions. The \emph{greedy algorithm} has cemented its place as a universal approach to arrive at approximate solutions in many application domains. In this algorithm, the set of agents is ordered (for instance, according to its index $i$) and then each agent sequentially solves the reduced optimization problem
\begin{equation} \label{eq:choice_gre}
    \gi_i \in \argmax_{a_i \in \ac_i} \W(\gi_1, \dots, \gi_{i-1}, a_i, \emp_{i+1}, \dots, \emp_n), 
\end{equation}
where each agent $i$ chooses the best action $\gi_i$ given that the previous agents in the sequence have also played their best action and the successive agents in the sequence have opted out. After each agent chooses according to Eq. \eqref{eq:choice_gre}, then the algorithm is complete and the resulting set of decisions $(\gi_1, \dots, \gi_n)$ comprises the joint decision set $\gi$. The process completes in $n$ time steps, where a time step is comprised of an agent making a decision and communicating that decision to future agents in the sequence.

However, the greedy algorithm makes a key assumption that agents have access to the decisions of the previous agents. In purely distributed systems, this assumption may be infeasible. There have been prior works that study the performance of the greedy algorithm with relaxed informational assumptions, in that agent $i$ only knows the decisions of some strict subset of the previous agents $S \subset \{1, \dots, i-1\}$ \cite{grimsman2017impact, Gharesifard2016a}. However, this work takes a different approach, where we assume that agents can make up for their informational deficiencies through a communication infrastructure. We model the communication constraints through an underlying graph structure $\gr = (\vr, \ed)$, where each vertex in $\vr$ corresponds to an agent in $\p$ and each edge $(i, j) \in \ed = \vr \times \vr$ implies that agents $i$ and $j$ can communicate with one another. The graph $\gr$ is assumed to be connected and undirected throughout this paper, unless explicitly stated. The set of agents that agent $i$ can communicate with is agent $i$'s neighborhood $\nd_i$. 

The primary focus of this work is to examine the interplay between the communication graph $\gr$ and the order in which the greedy algorithm in Eq. \eqref{eq:choice_gre} is solved under. For a given graph $\gr$, the order $\pi: \vr \to \p$ is defined by which label $i$ given to each vertex $v$. Therefore, given $\gr$, we would like the characterize the communication time guarantees of the worst order and the best order. To analyze the spectrum of possible guarantees with respect to different ordering methods, we define the following two quantities
\begin{align}
    \gtb &= \min_{\pi} \ \gt, \label{eq:btord} \\
    \gtw &= \max_{\pi} \ \gt, \label{eq:wtord}
\end{align}
where $\gt$ refers to the time it takes for the communication process to finish for a given graph $\gr$ and ordering $\pi$. We will use $\pib$ and $\piw$ to refer to the orderings that are the solutions of Eq. \eqref{eq:btord} and Eq. \eqref{eq:wtord} respectively. We remark that only in the full information setting, where $\gr_c$ is the complete graph, is the run-time for any order the same, with $T_{\max}(\gr_c) = T_{\min}(\gr_c) = n-1$.

We can describe the $k$-hop communication, in which an agent $i$'s greedy action $\gi_i$ is passed along to agents outside of its neighborhood $\nd_i$, using the following graph-theoretic notation. A \emph{walk} on the graph $\gr$ is a sequence of vertices $\gamma = (v_1, \dots, v_m)$, in which each successive pair $(v_j, v_{j+1}) \in \ed$ for all $1 \leq j < m$. We denote the length of the walk as $\card{\gamma}$ being the number of vertices in the sequence. A \emph{spanning walk} is a walk in which all vertices in the graph are visited and a \emph{minimum spanning walk} is a spanning walk with shortest length. A \emph{path} $p$ is a walk in which all the vertices $\{v_j\}_{j \leq m}$ are all distinct. The expression of $\gt$ is given as
\begin{equation}
    \label{eq:gtexp}
    \gt = \sum_{i = 1}^{n-1} \Big( \min_{\pathf} \card{\pathf} - 1 \Big),
\end{equation}
where $\pathf$ is a path on the graph from the vertex (labeled with) $i$ to the vertex $i+1$. This expression is motivated by a natural communication process, where initially agent $1$ computes its greedy action $\gi_1$ at time $0$. Then agent $1$ communicates $\gi_1$ to agent $2$ through a $k$-hop walk through the graph, where each hop is assumed to take $1$ time step. Then agent $2$ computes $\gi_2$ given $\gi_1$ and passes both actions to agent $3$ through another $k$-hop walk. Continuing this process, agent $n-1$ passes $\{\gi_j\}_{j < n}$ to agent $n$ and agent $n$ computes $\gi_n$ finishing the process. To isolate the run-time analysis with respect to only the communication time, we also assume that agents can solve for their greedy action $\gi_i$ arbitrarily fast.

\section{Main Results on Communication Time}
\label{sec:convergence}

\subsection{Motivating Example}
To make the communication process concrete, consider when the given communication graph is a line graph as shown in Figure \ref{fig:simcom}. In this graph scenario, agent $1$ initially computes its greedy response $\gi_1$ and passes the action it has played to agent $2$ at $t=0$. Then at $t=1$, agent $2$ (knowing $\gi_1$) can compute $\gi_2$ and passes both $\gi_1$ and $\gi_2$ along to agent $3$. Continuing this to $t=5$, agent $6$ will have been passed the greedy actions of all previous agents $1$ through $5$ and play its greedy action $\gi_6$, completing the greedy algorithm in Eq. \eqref{eq:choice_gre}. This will complete in $\gt=5$ time steps, which is the best that one can hope for when implementing the greedy algorithm in a limited information setting.

\begin{figure}[ht]
    \centering
    \includegraphics[width=\columnwidth]{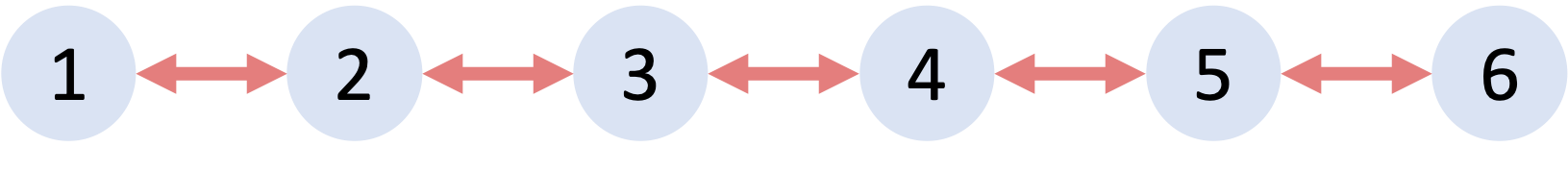}
    \caption{Example of a line graph, where  we have labeled the vertices according to the best ordering $\pib$. In this example, the agents compute their greedy action and pass it down the line.}
    \label{fig:simcom}
\end{figure}

However, consider the following order in Figure \ref{fig:simcom2}. This situation can occur if the order $\pi$ is improperly picked by the system operator. Under this ordering, agent $2$ can only receive the greedy action of agent $1$ through a $3$-hop path through agents $6$ and $4$, since there is not a direct communication link between agent $1$ and agent $2$. Following this logic, the greedy algorithm will complete at time $\gt = 3 + 5 + 4 + 3 + 2 = 17$, which can be seen to be significantly higher than the previous well chosen order.

\begin{figure}[ht]
    \centering
    \includegraphics[width=\columnwidth]{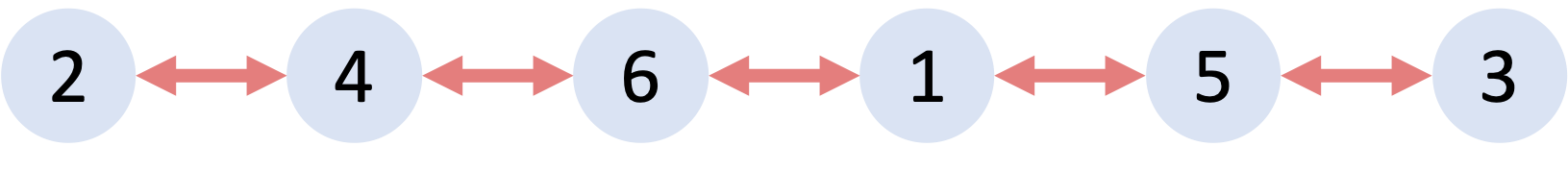}
    \caption{Example of a line graph, but instead we consider the adversarial ordering $\piw$ in which the vertices are labeled intermittently. In this instance, the $k$-hop communication path must bounce back and forth between agents to complete the greedy algorithm.}
    \label{fig:simcom2}
\end{figure}

Extending this argument to $n$ agents, under a line graph, the greedy algorithm under the best ordering $\pib$ will complete in $T(\gr,\pib) = n-1$ steps and under the worst ordering $\piw$ will complete in $T(\gr,\piw) = \floor{n^2/2}-1$ steps, where $\floor{a}$ is the floor function. Therefore, there may be a significant gap in the communication complexity that results from choosing different orderings. We analyze the possible gap by characterizing the quantities $\gtw$ and $\gtb$ in this paper.

\subsection{Communication Run-time Characterizations}
\label{subsec:worstbest}

We outline the main theorem of the paper below, where the worst case communication time over any graph structure is given for the best and worst orderings. The corresponding graph structures and orders that attain the worst-case communication time are displayed in Figure \ref{fig:simcom2} and Figure \ref{fig:starcom}.

\begin{figure}[ht]
    \centering
    \includegraphics[width=100pt]{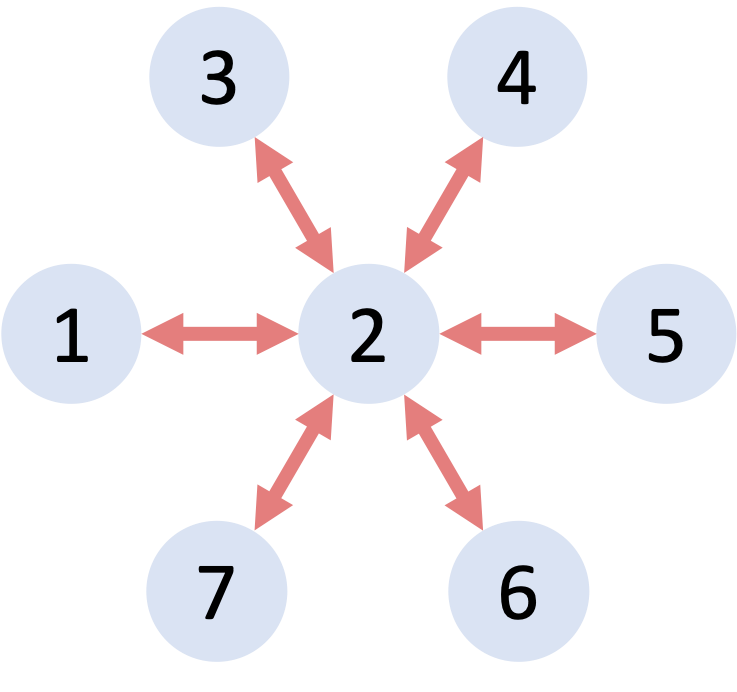}
    \caption{A $7$ node star graph with agent $2$ in the center. We note that the communication time $\gt$ on this graph using any ordering must be greater than $2 \cdot 7 - 4 = 10$. For $n \geq 3$ agents, the $n$-node star graph is the worst case graph for the best ordering $\pib$.}
    \label{fig:starcom}
\end{figure}

\begin{thm}
\label{thm:main}
Let $n \geq 3$ be the number of agents. The  maximum communication time required to complete the greedy algorithm in Eq. \eqref{eq:choice_gre} for any undirected, connected communication graph $\gr$ with the best and worst orderings is equal to
\begin{align}
    \max_{\gr} \ \gtb &=  2n - 4 \label{eq:gtworstb}\\
    \max_{\gr} \ \gtw &= \floor{n^2/2} - 1 \label{eq:gtworstw},
\end{align}
where $\gtb$ is defined in Eq. \eqref{eq:btord} and $\gtw$ is defined in Eq. \eqref{eq:wtord} and $\floor{a}$ is the largest integer that is below $a$.
\end{thm}
\begin{proof}
Proof is found in the Appendix.
\end{proof}

\subsection{Distributed Orderings that are Near-Optimal}
\label{subsec:disdes}

According to Theorem \ref{thm:main}, there is a significant complexity gap between using the best ordering $\pib$ and worst ordering $\piw$ for the communication time. However, finding the best order $\pib$ in general is not practical either due to computational restrictions or lack of information about the graph. So we would like to be able compute orderings that get as close to the run-time with $\pib$ as possible in a feasible manner. Therefore in this section, we construct an algorithm that can quickly compute a good ordering in conjunction with executing the greedy algorithm. An outline of the proposed algorithm is in Algorithm \ref{alg:distopt} with a fully distributed implementation found in \cite{konda2021}. The proposed design in essence computes a spanning walk on the graph that is close to the length of the minimum spanning walk through a variant of a depth-first search algorithm. 

The distributed implementation of Algorithm \ref{alg:distopt} to compute an approximate solution to Eq. \eqref{eq:opt} offers significant benefits over other distributed approaches. The communication scheme is simple, which allows for linear-time guarantees. This also means that the message complexity is low, where the bulk of the message is comprised of the previous agent's actions and the communication is robust to time delays. Lastly, since the base of Algorithm \ref{alg:distopt} is the greedy algorithm, we also inherit the corresponding performance guarantees. To be able to run Algorithm \ref{alg:distopt}, we assume that each agent (vertex) can store and access the following variables.
\begin{itemize}
    \item $v\mathrm{.actions} = \varnothing$ is the set of greedy actions that $v$ knows.
    \item $v\mathrm{.order} = \varnothing$ is the index in $\p$ that $v$ is labeled with.
    \item $v\mathrm{.parent} = \varnothing$ is $v$'s parent in the depth first search.
    \item $v\mathrm{.neighborhood}$ is the neighborhood set of $v$.
\end{itemize}

We also assume that a seed $v_{\mathrm{seed}}$ is given as the starting point of the Algorithm \ref{alg:distopt}. The communication time of Algorithm \ref{alg:distopt} is equivalent to the total number of calls to \textproc{message}, where the vertex $v$ messages either a vertex that hasn't been visited or its parent $v\mathrm{.parent}$. We keep track of the communication time through the variable $t$. The communication time guarantees of Algorithm \ref{alg:distopt} is given below.

\begin{prop}
    Let $n$ be the number of agents and $\talg$ be the output of Algorithm \ref{alg:distopt} given a communication graph $\gr$ and a seed vertex $v_{\mathrm{seed}} \in \vr$. The maximum communication time for any undirected, connected $\gr$ and seed $v_{\mathrm{seed}}$ is
    \begin{equation}
    \label{eq:algtime}
        \max_{\gr, v_{\mathrm{seed}}} \talg = 2n - 2
    \end{equation}
\end{prop}
\begin{proof}
Consider an arbitrary graph $\gr$ with $n$ agents and a seed vertex $v_{\mathrm{seed}}$. Let $\ell$ be the number of calls to \textproc{message} where a vertex $v$ messages another unvisited vertex and $m$ be the number of calls to \textproc{message} where a vertex $v$ messages its parent. It can be seen that Algorithm \ref{alg:distopt} will eventually visit all the vertices in the graph, so $\ell$ must equal $n-1$. Additionally, since $v_{\mathrm{seed}}$ does not have a parent and Algorithm \ref{alg:distopt} terminates if $v_{\mathrm{seed}}$ does not send a message to an unvisited neighbor, $m \leq \ell = n - 1$. Therefore for any $\gr$ and $v_{\mathrm{seed}}$ the communication time is bounded above by $\talg = m + \ell = 2n - 2$. Furthermore, it can be seen that for the star graph with $n$ vertices, $m = \ell = n - 1$, and the equality in Eq. \eqref{eq:algtime} is shown.
\end{proof}

Thus the communication guarantees of Algorithm \ref{alg:distopt} is only off by a constant of $2$ from the optimal communication guarantee of $2n - 4$ from the best ordering $\pib$. We remark that this difference can be further reduced if the termination condition is changed from `$v\mathrm{.parent}$ is not empty' to `$v\mathrm{.order} = n$', where $n$ is the number of agents.

\begin{algorithm}
\caption{Distributed Near-Optimal Ordering}
\label{alg:distopt}
\begin{algorithmic}
\Require graph $\gr$ and a vertex $v_{\mathrm{seed}} \in \vr$
\State \hspace{-1.3em} \textbf{Output:} time $t$
\State initialize the time $t \gets 0$
\State \Call{init}{$v_{\mathrm{seed}}$, $\varnothing$, $\varnothing$, 1}
\State \Return $t$ 
\Procedure{init}{$v$, $v_{\mathrm{par}}$, $\alpha$, $i$}
  \State label $v$ as $\mathrm{done}$
  \State update $v\mathrm{.order} \gets i$ and $v\mathrm{.parent} \gets v_{\mathrm{par}}$
  \State let $v$ compute $\gi_i$ from Eq. \eqref{eq:choice_gre} given actions $\alpha$
  \State update $v\mathrm{.actions} \gets \alpha \cup \{\gi_i\}$
  \State \Call{message}{$v$, $v\mathrm{.actions}$}
\EndProcedure

\Procedure{message}{$v$, $\alpha$}
\State update $v\mathrm{.actions} \gets \alpha$
  \If{exists $w$ in $v\mathrm{.neighborhood}$ not labeled $\mathrm{done}$}
  \State increment $t \gets t + 1$
  \State \Call{init}{$w$, $v$, $\alpha$, $\card{\alpha}$+1}
  \ElsIf{$v\mathrm{.parent}$ is not empty}
  \State increment $t \gets t + 1$
  \State \Call{message}{$v\mathrm{.parent}$, $\alpha$}
  \EndIf
\EndProcedure  
\end{algorithmic}
\end{algorithm}

\subsection{Directed Communication Graphs}

In this section, we consider the communication time guarantees with respect to the more general class of connected, directed graphs using different orderings. Under the class of undirected graphs, there is a significant gap in the communication time guarantees for the best $\pi_{\mathrm{best}}$ and the worst $\piw$ orderings. Not surprisingly, when we relax to the optimization problem $\max_{\gr_{\mathrm{dir}}} T_{\max}(\gr_{\mathrm{dir}})$ over the class of directed graphs, the worst case guarantees also increase. However, when considering the optimization problem for the best ordering over directed graphs $\max_{\gr_{\mathrm{dir}}} T_{\min}(\gr_{\mathrm{dir}})$, we have that the worst case guarantees are also of quadratic order. Therefore in directed graphs, the gap between the performance guarantees under different orderings is relatively small.

\begin{figure}[ht]
    \centering
    \includegraphics[width=\columnwidth]{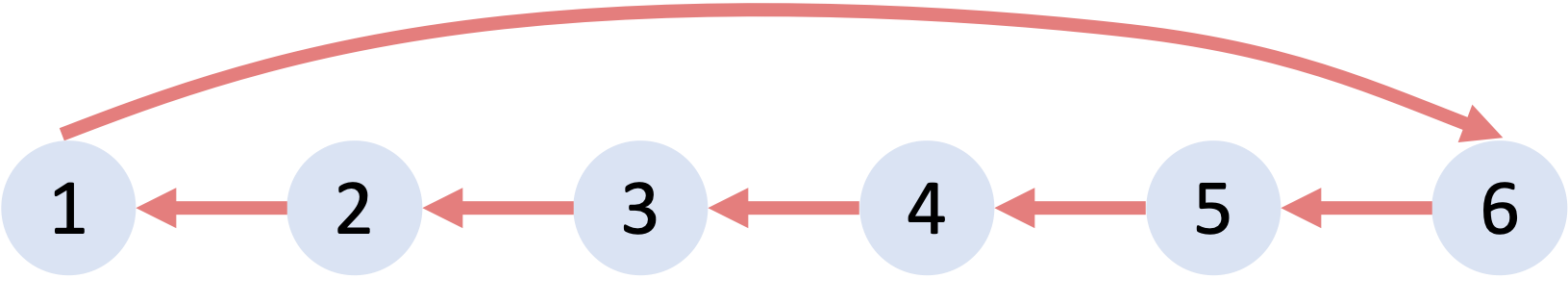}
    \caption{Example of a directed cycle graph, where the vertices are labeled adversarially with $\piw$. Here, the $k$-hop communication must cycle back to get to the next agent.}
    \label{fig:dircom1}
\end{figure}

For the graph example in Figure \ref{fig:dircom1} using the worst ordering $\piw$, notice that to get from $i$ to $i+1$, every edge but one in the directed graph must be traversed, resulting in a communication time of $5 \times 5$ for $6$ agents. For the graph example in Figure \ref{fig:dircom2} using any ordering $\pi$, the vertices labeled with $2$, $3$, and $4$ must be traversed every time to reach the vertices labeled with $5$, $6$, $7$, $8$ in order, starting from the vertex labeled $1$. Thus the communication time for the graph under the best ordering is $4 \times 4$ for $8$ agents. Extending these constructions to $n$ agents, we arrive at the following lemma.

\begin{figure}[ht]
    \centering
    \includegraphics[width=150pt]{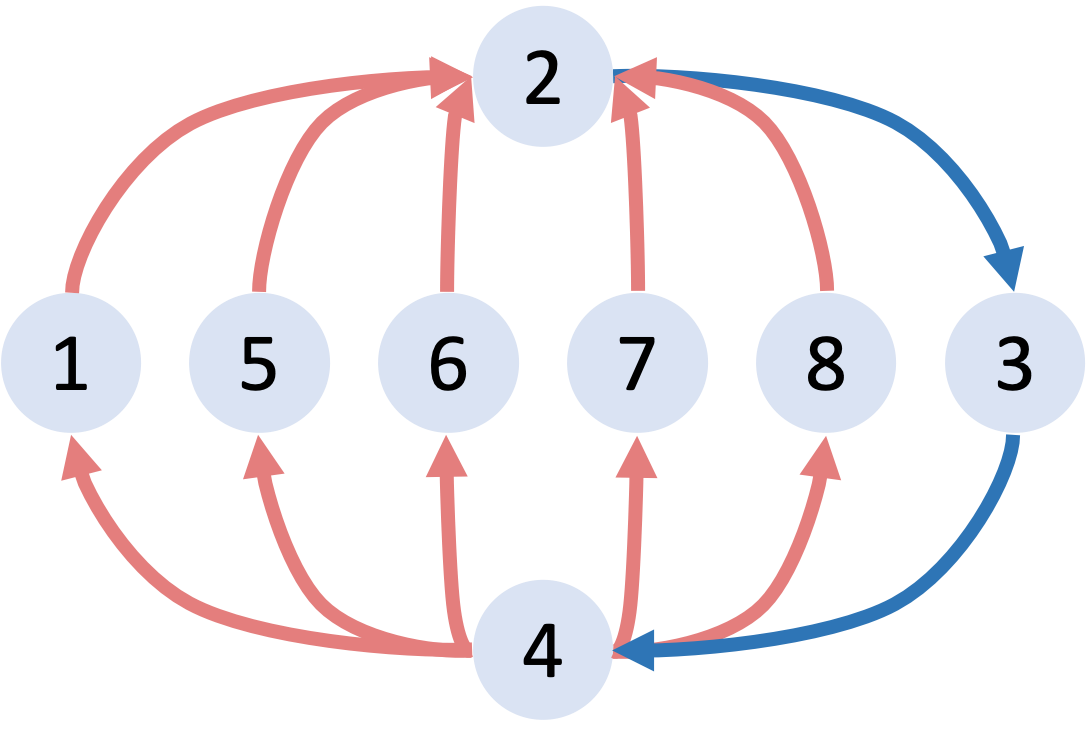}
    \caption{Example of a directed graph that has the worst communication complexity for the best ordering $\pib$. Here, the vertices in the latter half must cycle back to get to the next labeled vertex.}
    \label{fig:dircom2}
\end{figure}

\begin{prop}
\label{prop:worstdir}
Let $n \geq 3$ be the number of agents. The  maximum communication time required to complete the greedy algorithm in Eq. \eqref{eq:choice_gre} for any directed, connected communication graph $\gr_{\mathrm{dir}}$ with the best and worst ordering is
\begin{align}
    \max_{\gr_{\mathrm{dir}}} \ T_{\min}(\gr_{\mathrm{dir}}) &= \floor{\frac{n}{2}} \cdot \ceil{\frac{n}{2}}, \label{eq:gtworstbdir}\\
    \max_{\gr_{\mathrm{dir}}} \ T_{\max}(\gr_{\mathrm{dir}}) &= (n-1)^2 \label{eq:gtworstwdir},
\end{align}
where $\floor{a}$ is largest integer smaller than $a$ and $\ceil{a}$ is the smallest integer larger than $a$.
\end{prop}
\begin{proof}
Proof is found in the Appendix.
\end{proof}

\section{Simulations}
\label{sec:sim}

We analyze our theoretical results for the communication time guarantees empirically through a simulation, presented in Figure \ref{fig:graph1} and Figure \ref{fig:graph2}. The code used to generate the simulations can be found in \cite{konda2021}. We use the model of Erdos-Renyi networks \cite{erdos1960evolution}, where each possible undirected pair of edges $(i, j)$ has a probability $P \leq 1$ of existing, to generate a sample set of possible graph structures. For Figure \ref{fig:graph1}, we sample $200$ instances of Erdos-Renyi networks with $6$ nodes and a probability parameter of $P = .3$. For each graph, we calculate the communication time for the greedy algorithm using the best ordering $\pib$, the ordering given by Algorithm \ref{alg:distopt}, and a randomly assigned ordering. For Figure \ref{fig:graph2}, we sample $300$ instances of Erdos-Renyi networks with $40$ nodes and a probability parameter of $P = .05$. In Figure \ref{fig:graph2}, we calculate the communication time for only the ordering given by Algorithm \ref{alg:distopt} and a randomly assigned ordering.

We observe in Figure \ref{fig:graph1} that indeed the best ordering $\pib$ achieves the lowest distribution of communication times, centered closely to $n=6$. The distribution of communication times of the ordering given in Algorithm \ref{alg:distopt} is noticeably close to the one of $\pib$, with indeed no run-times over $2n = 12$. The distribution of the communications using random orderings does perform the worst with the largest spread. In Figure \ref{fig:graph2}, we compare communication times from the ordering from Algorithm \ref{alg:distopt} directly with the random ordering, as computing the best ordering $\pib$ is infeasible for large $n$. We see the same trends in Figure \ref{fig:graph1} reflected in a more extreme fashion. The distribution of communication times using the ordering of Algorithm \ref{alg:distopt} is still upper bounded by $2n = 80$. However, the distribution times of communication times using the random ordering is now centered much higher with a larger spread as well. Therefore, we see significant benefits from using the ordering from Algorithm \ref{alg:distopt} rather than the naive approach of using random ordering.

\begin{figure}[ht]
    \centering
    \includegraphics[width=\columnwidth]{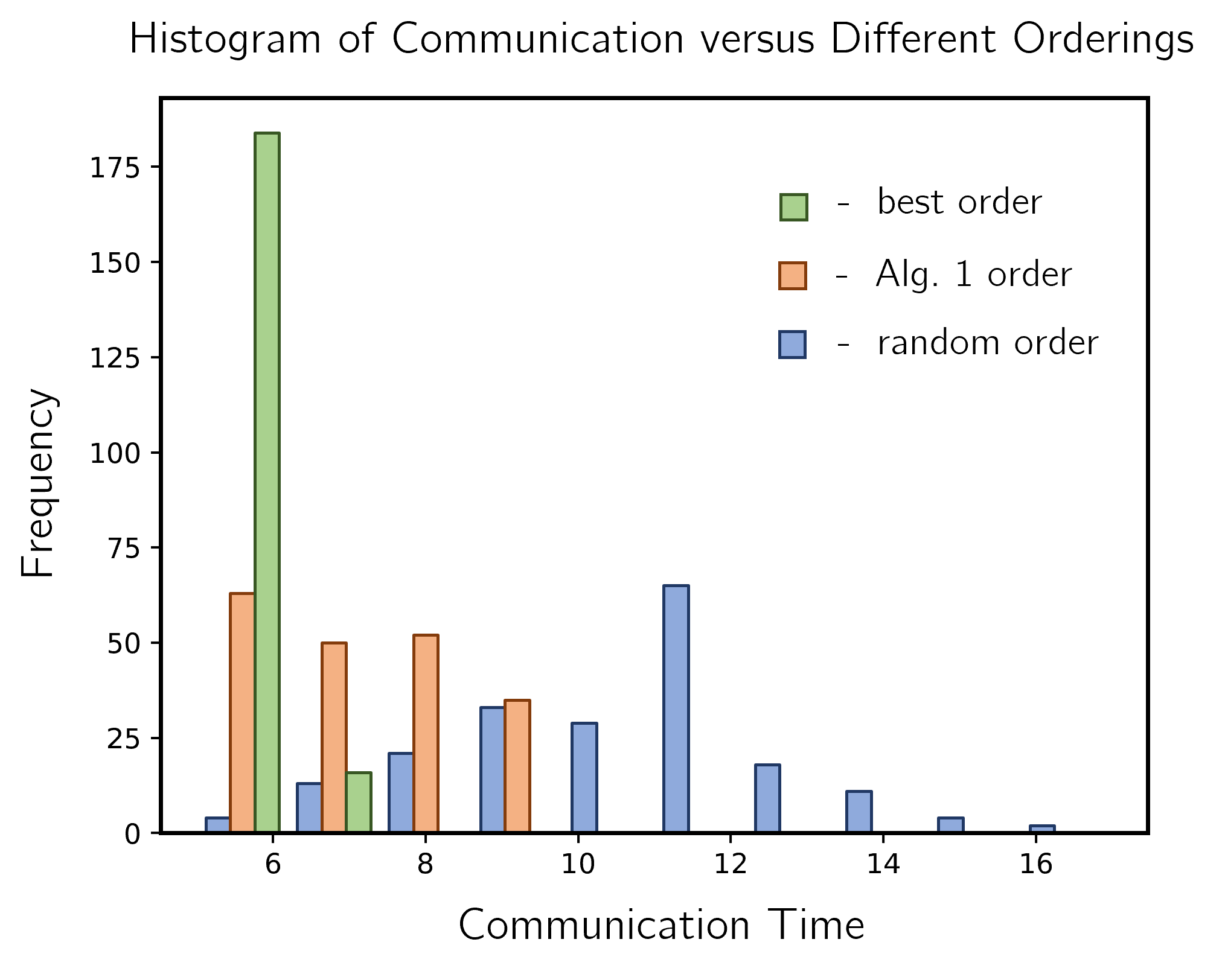}
    \caption{We show the distribution over of communication times needed to complete the greedy algorithm for $200$ instances of random graphs generated by a Erdos-Renyi process with respect to the random ordering, the best ordering $\pib$, and the ordering from Algorithm \ref{alg:distopt}. For the graph parameters $n=6$ number of agents and $P=.3$ probability of an edge existing, we see that $\pib$ gives slightly lower average communication times than the ordering from Algorithm \ref{alg:distopt}, but both offer significant improvements over the random ordering.}
    \label{fig:graph1}
\end{figure}

\begin{figure}[ht]
    \centering
    \includegraphics[width=\columnwidth]{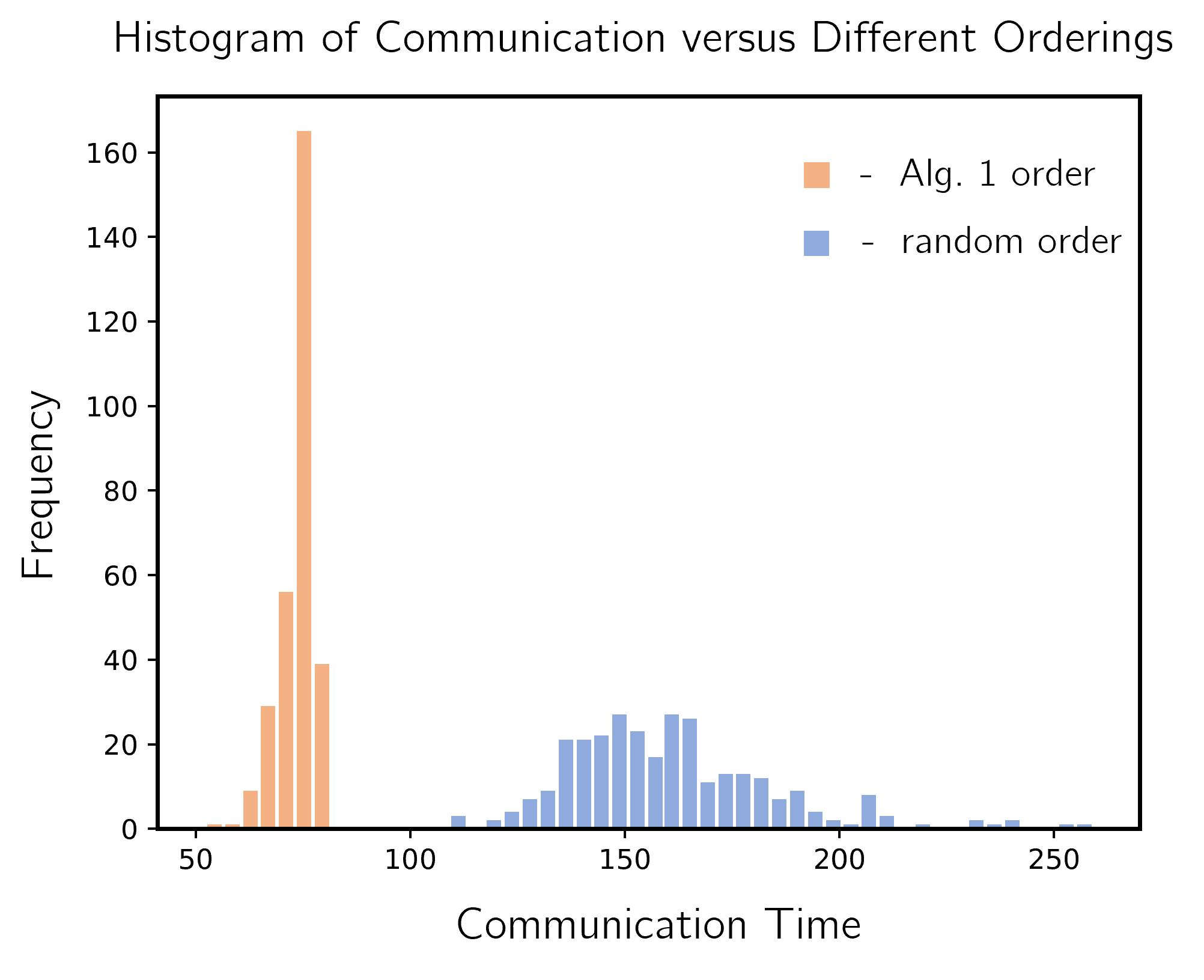}
    \caption{We compare the communication times under the random ordering to the times under the ordering algorithm proposed in Algorithm \ref{alg:distopt}. For the graph parameters $n=40$ number of agents and $P=.05$ probability of an edge existing, we see a marked decrease in the communication time from the random ordering to using the proposed algorithm over a set of $300$ randomly generated graphs.}
    \label{fig:graph2}
\end{figure}

\section{Submodular Maximization}
\label{sec:AA}

In this section, we discuss submodular maximization problems, which can be modeled as multiagent decision problems. Consider a base set of elements $E$, and let $a_i \subseteq E$, $\ac_i \subseteq 2^E$, and $\emp_i = \emptyset$. The objective function takes the form $W(a) = f\left(\cup_{a_i \in a} a_i\right)$, where $f:2^E \to  \mathbb{R}$ has the following properties for any $A \subseteq B \subseteq E$:
\begin{enumerate}
    \item \emph{Submodular}: $f(A \cup \{x\}) - f(A) \ge f(B \cup \{x\}) - f(B)$ for all $x \in E \setminus B$
    \item \emph{Monotone}: $f(A) \le f(B)$
    \item \emph{Normalized}: $f(\emptyset) = 0$
\end{enumerate}
In this setting, it has been shown that the greedy algorithm that is implemented in Algorithm \ref{alg:distopt} guarantees that $W(\gi) \ge (1/2) W(\aopt)$, where $\gi$ is defined in Eq. \eqref{eq:choice_gre}. 

\subsection{Comparison with Other Distributed Algorithms}

As mentioned previously, much work has been done to develop other algorithms to solve submodular maximization. For instance, \cite{rezazadeh2021distributed} presents a similar distributed algorithm, using a multilinear extension, and a distributed pipage rounding technique. At each time step, each agent performs a calculation for each action based on a sample of $K$ actions drawn from a probability distribution. After $T$ time steps, the performance guarantee is $(1 - 1/e)(1 - (2d(\gr)n + n/2 + 1)(n/T))$ with probability at least $1 - 2nTe^{-K/(8T^2)}$. Thus, for high $T$ and $K = O(T^2)$, there is a high probability that the algorithm gives the $1-1/e$ guarantee. Using this information, the algorithm could provide a $1/2$ guarantee only for $T \ge 4.78(2d(\gr)n^2 + n^2/2 + 1)$, and only with high probability when $K = O(T^2)$.

The paper \cite{du2020jacobi} describes a Jacobi-style algorithm, where at each time step agent $i$ creates a strategy profile, i.e., a probability distribution across each of its actions. Then, it chooses $K$ of those values to share with its neighbors to propagate through the network. It was shown that the resulting decision set approaches being within $1/2$ the optimal as the number of iterations increases. It is only shown in the paper that the probability of achieving the $1/2$ guarantee is $1-O(1/T)$ rather than an explicit time expectation. However, the examples in the paper suggest that it may take $T \ge n^2$ or more time steps to realize this.

In another example, \cite{Robey2019a} presents the Constraint-Distributed Continuous Greedy (CDCG), a consensus-style algorithm, in which agent $i$ shares an $m$-vector with all its neighbors at each time step, where $m$ is the number of actions available to $i$. It is shown that the resulting decision set approaches being within $1-1/e$ of the optimal as the number of iterations $T$ increases. The error in the performance guarantee vanishes at a rate of $O(n^{5/2}/T)$, and therefore, it may require $T \ge n^{5/2}$ time steps in order to reach an acceptable error.

In each of the three methods listed above, each time step requires each agent do perform some calculation for each of its actions. The time requirement for each to reach an acceptable solution is expected to be greater than $2n-2$, which is the number of time steps it takes to complete Algorithm~\ref{alg:distopt}. This suggests that there is a tradeoff between performance guarantees and time complexity: Algorithm~\ref{alg:distopt} achieves the 1/2 guarantee quickly, but other algorithms converge to a solution within $1 - 1/e$, but more slowly \footnote{Although we do not present a rigorous analysis here, we assert that Algorithm~\ref{alg:distopt} also requires less information exchange at each time step. This will be a topic of future work.}.

\section{Conclusion}
\label{sec:conc}
In this work, we analyze the greedy algorithm in a multi-agent context. More specifically, when the agents have limited information about the other agents but can communicate with a subset of the other agents, we characterize the effects of using different agent orderings on the communication time over the network. First, we have shown that under the best ordering over any graph, the communication time is $2n-4$, where $n$ is the number of agents. When considering the worst ordering, we also show that this bound increases to $\floor{n^2/2}-1$ holding for any communication network. We also provide an algorithm that can be used in a distributed fashion to obtain a communication time guarantee of $2n-2$, which is a constant factor away from the best communication guarantee. Furthermore, we verify these results computationally in Section \ref{sec:sim} and describe the implications in the context of submodular optimization problems in Section \ref{sec:AA}. Future work is comprised of extending this work to analyze the interplay between agent orderings, the communication time guarantees, and possible performance guarantees.

\bibliographystyle{ieeetr}
\bibliography{references.bib}

\appendix

\begin{proof}[Proof of Theorem \ref{thm:main}]
We first show the equality in Eq. \eqref{eq:gtworstb}. For a given graph $\gr$ and the optimal ordering $\pib$, we claim that the communication time is equal to
\begin{equation}
    \label{eq:minspan}
    \gtb = \card{\gamma_{\min}},
\end{equation}
where $\gamma_{\min}$ is a minimum spanning walk of $\gr$. Note that for any given ordering $\pi$, the communication time for $\gt$ is given in Eq. \eqref{eq:gtexp}. Let $\pathf^*$ be the shortest path from $i$ to $i+1$. The walk $\gamma_{\mathrm{cat}} = p^*_{1 \to 2}p^*_{2 \to 3} \dots p^*_{n-1 \to n}$ is defined as the concatenation of the shortest paths from $1$ to $n$ with the duplicate vertices from $\pathf^*$ and $p_{i+1 \to i+2}^*$ removed. Then according to Eq. \eqref{eq:gtexp}, we have that $\gt = \card{\gamma_{\mathrm{cat}}}$. We note that since $\pi^{-1}(i)$ and $\pi^{-1}(i+1)$ are in $\pathf^*$, then $\gamma_{\mathrm{cat}}$ is a spanning walk and thus $\card{\gamma_{\mathrm{cat}}} \geq \card{\gamma_{\min}}$. Since the ordering $\pi$ was arbitrary and $\gt = \card{\gamma_{\mathrm{cat}}}$, we have that 
$\min_{\pi} \gt \geq \card{\gamma_{\min}}$.
This expression actually holds with equality, as $\pib$ can be taken as the order that the vertices first appear in $\gamma_{\min}$, matching Eq. \eqref{eq:minspan} and the claim is shown.

Notice that a spanning walk for the graph $\gr_1 = (\vr_1, \ed_1)$ is also a spanning walk for the graph $\gr_2 = (\vr_1, \ed_1 \cup \{e\})$ with the added edge $e$. Then the length of the minimal spanning walk $\card{\gamma^{1}_{\min}}$ for $\gr_1$ must be at least $\card{\gamma^{2}_{\min}}$ for $\gr_{2}$. Thus $T_{\min}(\gr_1) \geq T_{\min}(\gr_2)$, and to calculate $\max_{\gr} \gtb$, it is sufficient to restrict to the class of tree graphs, which is the class of connected graphs with the least number of edges.

Consider any spanning walk $\gamma$ on the tree graph $\gr_{\T} = (\vr_{\T}, \ed_{\T})$ starting at the vertex $v_1$ and ending at the vertex $v_n$. Let $\ed^p_{\T} \subset \ed_{\T}$ be the set of edges that belong to the unique path $p$ between $v_1$ and $v_n$. We claim that $\gamma$ must visit each edge $e \in \ed_{\T}$ at least once and every edge $e \in \ed_{\T} \setminus \ed_{\T}^p$ at least twice. If there exist an edge $\hat{e} \in \ed_{\T}$ such that $\hat{e} \notin \gamma$, then since $\gamma$ is a spanning walk, then the graph $\hat{\gr} = (\vr_{\T}, \ed_{\T} \setminus \{ \hat{e} \})$ must also be a connected graph. But this is a contradiction, since $\gr_{\T}$ is assumed to be a tree graph. If an edge $\bar{e} \in \ed_{\T} \setminus \ed_{\T}^p$ is removed from $\gr_{\T}$, there must be two nonempty components, one containing both $v_1$ and $v_n$ and another containing neither. Thus if $\bar{e}$ is only traversed once in the walk, it must hop from $v_1$ from the first component to the other component once. However, the spanning walk $\gamma$ cannot come back to the first component again, contradicting our definition of $v_n$ and the claim is shown.

Therefore for a tree graph $\gr_{\T}$ and any spanning walk $\gamma$, we have that $\card{\gamma} \geq 2 \card{\ed_{\T}} - \card{\ed^p_{\T}}$. Moreover there exists a spanning walk that has the length equal to $\card{\gamma} = 2 \card{\ed_{\T}} - \card{\ed^p_{\T}}$ in which the vertices that are not along the path $p$ from $v_1$ to $v_n$ are reached through a cycle that visits each edge not part of $p$ twice. Thus, the length of the minimum spanning walk can be written as the following optimization problem. Here, $\mathrm{diam}(\gr_{\T})$ is the diameter of the graph $\gr_{\T}$ and $\Gamma_{\gr_{\T}}$ is the set of spanning walks.

\begin{align*}
    \card{\gamma_{\min}} &= \min_{\gamma \in \Gamma_{\gr_{\T}}} \card{\gamma} \\
    &= 2 \card{\ed_{\T}} - \max_{v_1, v_n} \card{\ed^p_{\T}} = 2(n-1) - \mathrm{diam}(\gr_{\T}).
\end{align*}

For a given tree graph $\gr_{\T}$ with more than $n \geq 3$ vertices, the diameter must be greater than $\mathrm{diam}(\gr_{\T}) \geq 2$. Therefore, the length of the minimum spanning walk must be less than $2n - 4$ for any tree graph $\gr_{\T}$. Moreover, the star graph is the tree graph with a graph diameter of $2$, so $\max_{\gr} \gtb = 2n - 4$.

Now, we show the equality in Eq. \eqref{eq:gtworstw}. We claim that the connected, undirected graph that attains $\max_{\gr} \gtw$ is the line graph. We observe, similarly as before, that if a path $\pathf$ exists from $i$ to $i+1$ for the graph $\gr_1 = (\vr_1, \ed_1)$, then it must also exist for the graph $\gr_2 = (\vr_1, \ed_1 \cup \{e\})$ with the added edge $e$ for any $1 \leq i \leq n-1$. As the run-time in Eq. \eqref{eq:gtexp} is defined by the shortest path from $i$ to $i+1$, the run-time for $\gr_2$ is lower bounded by $T(\gr_1, \pi) \geq T(\gr_2, \pi)$. Thus, we can assume that worst-case graph is a tree graph without loss of generality. If $\gr_{\T}$ is a tree graph, then the path $\pathf$ from $\pi^{-1}(i)$ to $\pi^{-1}(i+1)$ is unique.

We now claim that for any tree graph $\gr_{\T}$ for some ordering $\pi$, there exists an ordering $\pi_L$ with the line graph that achieves at least the same run-time. If $\gr_{\T}$ is a tree graph that is not the line graph, there exists at least one vertex $v_{c}$ of $\gr_{\T}$ with degree $d \geq 3$. Let $\{v_j\}_{1 \leq j \leq d}$ be the vertices in the neighborhood of $v_c$. Also, let $\T_j \subset \gr_{\T}$ be the corresponding tree component containing $v_j$ that results from removing the edge $(v_j, v_c)$. Let $v_j^{\mathrm{d}} \in \T_j$ be the vertex farthest away from $v_j$ and $\mathrm{d}_j$ be the distance from $v_j$ to $v_j^{\mathrm{d}}$. Additionally let $H = \{(\pi^{-1}(i), \pi^{-1}(i+1))\}_{i < n} $ and
\begin{equation*}
H^j = \{(v, v') \in H : v \in \T_1, v' \in \T_j \text{ or } v \in \T_j, v' \in \T_1 \}.    
\end{equation*}
We also denote that $(v, v_c)$ and $(v_c, v)$ are included in $H^1$. Consider the modified graph $\gr^J_{\T}$, where the edge $(v_{j=1}, v_c)$ is replaced with $(v_{j=1}, v_{j=J}^{\mathrm{d}})$ for some $1 < J \leq d$. Then the communication time for the graph $\gr^J_{\T}$ is
\begin{align*}
     T(\gr^J_{\T}, \pi) &= \sum_{(v, v') \in H} p^{\gr^J_{\T}}_{v, v'} =  \sum_{j \geq 1} \sum_{(v, v') \in H^j } p^{\gr^J_{\T}}_{v, v'} \\
     &\geq \sum_{(v, v') \in H} p^{\gr_{\T}}_{v, v'} + \sum_{j > 1} \sum_{(v, v') \in H^j } \mathrm{d}_J - 2 \sum_{(v, v') \in H^J} \mathrm{d}_J \\
     &\geq \gt \text{ for some } J.
\end{align*}
Here, $p^{\gr^J_{\T}}_{v, v'}$ refers to the path from $v$ to $v'$ in $\gr^J_{\T}$. The first equality comes from rewriting Eq. \eqref{eq:gtexp} using $H^j$. The last inequality comes from the fact that the degree $d \geq 3$ and that 
\begin{equation*}
\max_{J > 1} \Big\{ \sum_{j > 1} \sum_{(v, v') \in H^j } \mathrm{d}_J - 2 \sum_{(v, v') \in H^J} \mathrm{d}_J \Big\} \geq 0.
\end{equation*}
By applying a similar argument inductively for every tree $\T_j$ for $j \leq d-2$ and for every vertex $v_c$ with degree more than $3$, we have the claim.

The worst ordering for the line graph and the corresponding guarantee of $\floor{n^2/2} - 1$ is given by the work in \cite[Theorem 8]{bulteau_et_al:LIPIcs.CPM.2021.11}.
We give a sketch of the proof for completeness. Let $v_i$ be the $i$'th vertex in the line graph. Consider any ordering $\pi$. If $\pi$ has any of the following properties, then there must exist another ordering $\pi'$ that has a larger communication time that results from swapping positions of some $\pi(v)$ and $\pi(v')$.
\begin{enumerate}
    \item For some $1 < i < n$, either $(\pi(v_1), \pi(v_i), \pi(v_{i+1}))$ or $(\pi(v_i), \pi(v_{i+1}), \pi(v_{n}))$ is monotonic.
    \item For some $i, j < n$, there are two pairs $(\pi(v_i), \pi(v_{i+1}))$ and $(\pi(v_j), \pi(v_{j+1}))$ that are separated by some threshold $1 \leq \ell \leq n$.
    \item There is a triple $(\pi(v_i), \pi(v_{i+1}), \pi(v_{i+2}))$ that is monotonic.
\end{enumerate}
If an ordering $\pi$ does not have any of the previous three  properties, then it must have all even indexed vertices $\{v_i\}_{i \ \mathrm{even}}$ below the threshold $\floor{n/2}$ and all odd indexed vertices $\{v_i\}_{i \ \mathrm{odd}}$ above the threshold with the middle vertices corresponding to the endpoints $1$ and $n$. An example of this configuration is shown in Figure \ref{fig:simcom2}. It can be seen that the line graph with this configuration has a $\gt = \floor{n^2/2} - 1$.
\end{proof}

\begin{proof}[Proof of Proposition \ref{prop:worstdir}]

First, extending the directed cycle graph $\mathcal{C}_n$ in Figure \ref{fig:dircom1} to $n$ agents with the given worst case labeling $\pi_{\rm{worst}}$ produces a lower bound 
\begin{equation*}
\max_{\gr_{\mathrm{dir}}} T_{\max}(\gr_{\mathrm{dir}}) \geq T(\mathcal{C}_n, \pi_{\rm{worst}}) = (n-1)^2.
\end{equation*}
The upper bound $\max_{\gr_{\mathrm{dir}}} T_{\max}(\gr_{\mathrm{dir}}) \leq (n-1)^2$ also holds, since for any graph $\gr$ and ordering $\pi$, the communication time must satisfy $\gt \leq \sum_{i=1}^{n-1}(n-1)$ as the length of any path in the graph cannot be greater than $n-1$.

Now we show that $\max_{\gr_{\mathrm{dir}}} T_{\min}(\gr_{\mathrm{dir}}) \geq \floor{\frac{n}{2}} \cdot \ceil{\frac{n}{2}}$ by extending the graph construction, denoted $\mathcal{D}_n$, in Figure \ref{fig:dircom2} to $n$ agents. Formally, the edge set of $\mathcal{D}_n$ includes $(v_j, v_{j+1})$ for every $1 \leq j \leq \ceil{n/2}-2$ as well as the edges $(v_{\ceil{n/2}-1}, v_j)$ and $(v_j, v_{1})$ for every $\ceil{n/2} \leq j \leq n$. We confirm that $T(\mathcal{D}_n, \pi_{\rm{best}}) = \floor{\frac{n}{2}} \cdot \ceil{\frac{n}{2}}$ for the best order. Consider any ordering $\pi$ and without loss of generality, assume that $\pi(v_j) < \pi(v_{j+1})$ for all $\ceil{n/2} \leq j \leq n$. Thus the communication time has to be lower bounded by
\begin{equation*}
    T(\mathcal{D}_n, \pi) \geq \sum_{j \geq \ceil{n/2}}^{n-1} \min_{p_{v_j \to v_{j+1}}} \card{p_{v_j \to v_{j+1}}} - 1 \geq \floor{\frac{n}{2}} \cdot \ceil{\frac{n}{2}},
\end{equation*}
as $(v_j, v_1, v_2, \dots, v_{\ceil{n/2}-1}, v_{j+1})$ is the unique path $p_{v_j \to v_{j+1}}$ from $v_j$ to $v_{j+1}$ with a length of $\ceil{n/2} + 1$. Observe that the ordering $\pi(v_j) = j + 1$ for $1 \leq j \leq n-1$ and $\pi(v_n) = 1$ achieves this communication time and so the lower bound is indeed tight.

Now we show that for any given directed graph $\gr$, there exists an ordering $\hat{\pi}$ in which $T(\gr, \hat{\pi}) \leq \floor{\frac{n}{2}} \cdot \ceil{\frac{n}{2}}$. Let $p_{\rm{long}} = (v_1, \dots, v_l)$ be the longest path of the graph, where $l = \card{p_{\rm{long}}}$. If $l = n$, then $p_{\rm{long}}$ is a spanning walk on the graph and the ordering $\pi(v_j) = j$ along the longest path results in a communication time of $T(\gr, \hat{\pi}) = n-1 \leq \floor{\frac{n}{2}} \cdot \ceil{\frac{n}{2}}$. Otherwise, since $\gr$ is assumed to be strongly connected there exists a vertex $v_J$ in $p_{\rm{long}}$ that is adjacent to another vertex $\bar{v}$ that is not in the path $p_{\rm{long}}$. We construct the ordering $\hat{\pi}$ as follows. The vertices in $p_{\rm{long}}$ are labeled as $\hat{\pi}(v_j) = j$ for $j \leq J$ and $\hat{\pi}(v_j) = n - l + j$ for $j > J$. The vertex $\bar{v}$ is labeled with $\hat{\pi}(\bar{v}) = J + 1 $ and the labels for the rest of the vertices in $\gr$ can be arbitrarily selected from $\{J+2, \dots, n - l + J \}$. The resulting communication time along this order, using Eq. \eqref{eq:algtime}, is
\begin{align*}
    T(\gr, \hat{\pi}) &= \sum_{i = 1}^{n-1} \Big( \min_{\pathf} \card{\pathf} - 1 \Big) \\
    &= (l - 1) + \sum_{i = J+1}^{n-l+J} \Big( \min_{\pathf} \card{\pathf} - 1 \Big),
\end{align*}
as the vertices $\hat{\pi}^{-1}(i)$ and $\hat{\pi}^{-1}(i+1)$ are adjacent to each other if $1 \leq i \leq J$ or $n - l + J + 1 \leq i \leq n - 1$ according to the prescribed order $\hat{\pi}$. For any $i \in \p$, we observe that $\min_{\pathf} \card{\pathf} \leq l$ by definition of $p_{\rm{long}}$. Now we have the upper bound
\begin{equation*}
    T(\gr, \hat{\pi}) \leq (l - 1) + (n - l)(l - 1) \leq \floor{\frac{n}{2}} \cdot \ceil{\frac{n}{2}},
\end{equation*}
where the expression achieves the maximum at $l = \ceil{n/2} + 1$. Thus, since $\gr$ was arbitrary chosen,
\begin{equation*}
    \max_{\gr_{\mathrm{dir}}} T_{\min}(\gr_{\mathrm{dir}}) = \max_{\gr_{\mathrm{dir}}} T(\gr, \hat{\pi}) \leq \floor{\frac{n}{2}} \cdot \ceil{\frac{n}{2}},
\end{equation*}
and thus we have shown equality.
\end{proof}

\end{document}